\documentclass[12pt]{article}

\usepackage[margin=1in]{geometry}  
\usepackage{graphicx}              
\usepackage{amsmath}               
\usepackage{amsfonts}              
\usepackage{amsthm}                
\usepackage{amssymb}
\usepackage{mathrsfs}
\usepackage{enumerate}
\usepackage{algorithmic}
\usepackage{url}
\usepackage{seqsplit}

\theoremstyle{plain}
\newtheorem{theorem}{Theorem}

\newtheorem{proposition}[theorem]{Proposition}

\theoremstyle{definition}

\algsetup{linenodelimiter=.}

\newcommand{\refsection}[1]{Section \ref{#1}}
\newcommand{\refequation}[1]{Equation (\ref{#1})}

\newcommand{\refalgorithm}[1]{Algorithm \ref{#1}}

\newcommand{\Romnum}[1]{\uppercase\expandafter{\romannumeral #1}}

\def\F{\mathbb{F}}
\def\K{\mathbb{K}}
\def\M{\mathsf{M}}
\def\CC{\mathsf{C}}

\makeatletter
\newcounter{algorithm}
\renewcommand{\thealgorithm}{\arabic{algorithm}}
\def\algorithm{\@ifnextchar[{\@algorithma}{\@algorithmb}}
\def\@algorithma[#1]{%
  \refstepcounter{algorithm}
  \trivlist
  \leftmargin\z@
  \itemindent\z@
  \labelsep\z@
  \item[\parbox{\columnwidth}{%
    \hrule
    \hrule
    \noindent\strut\textbf{Algorithm \thealgorithm} #1
    \hrule
  }]\hfil\vskip0em%
}
\def\@algorithmb{\@algorithma[]}

\makeatother

\newcommand{\keywords}[1]{\begin{description} \item[\textbf{Keywords.}] #1 \end{description}}

\title{Taking Roots over High Extensions \\ of Finite Fields}

\author{Javad Doliskani \\ \texttt{jdoliska@uwo.ca} 
\and
\'{E}ric Schost \\ \texttt{eschost@uwo.ca}}

\date{}

\begin{document}

\maketitle

\begin{abstract}
We present a new algorithm for computing $m$-th roots over the finite
field $\F_q$, where $q = p^n$, with $p$ a prime, and $m$ any positive
integer. In the particular case $m=2$, the cost of the new algorithm
is an expected $O(\M(n)\log (p) + \CC(n)\log(n))$ operations in
$\F_p$, where $\M(n)$ and $\CC(n)$ are bounds for the cost of
polynomial multiplication and modular polynomial composition. Known
results give $\M(n) = O(n\log (n) \log\log (n))$ and $\CC(n) =
O(n^{1.67})$, so our algorithm is subquadratic in $n$.  
\end{abstract}


\section{Introduction}\label{section:intro}

Beside its intrinsic interest, computing $m$-th roots over finite
fields (for $m$ an integer at least equal to $2$) has found many
applications in computer science. Our own interest comes from elliptic
and hyperelliptic curve cryptography; there, square root computations
show up in pairing-based cryptography~\cite{BaKiLySc02} or
point-counting problems~\cite{GaSc10}.

Our result in this paper is a new algorithm for computing $m$-th roots
in a degree $n$ extension $\F_q$ of the prime field $\F_p$, with $p$ a
prime. Our emphasis is on the case where $p$ is thought to be small,
and the degree $n$ grows. Roughly speaking, we reduce the problem to
$m$-th root extraction in a lower degree extension of $\F_p$ (when
$m=2$, we actually reduce the problem to square root extraction over
$\F_p$ itself).


\paragraph{Our complexity model.}
It is possible to describe the algorithm in an abstract manner,
independently of the choice of a basis of $\F_q$ over $\F_p$. However,
to give concrete complexity estimates, we have to decide which
representation we use, the most usual choices being monomial and
normal bases. We choose to use a monomial basis, since in particular
our implementation is based on the library NTL~\cite{NTL2009}, which
uses this representation.  Thus, the finite field $\F_q=\F_{p^n}$ is
represented as $\F_p[X]/\langle f\rangle$, for some monic irreducible
polynomial $f \in \F_p[X]$ of degree $n$; elements of $\F_q$ are
represented as polynomials in $\F_p[X]$ of degree less than $n$. We
will briefly mention the normal basis representation later on.

The costs of all algorithms are measured in number of operations
$+,\times,\div$ in the base field $\F_p$ (that is, we are using an
algebraic complexity model).

We shall denote upper bounds for the cost of {\em polynomial
  multiplication} and {\em modular composition} by respectively
$\M(n)$ and $\CC(n)$. This means that over any field $\K$, we can
multiply polynomials of degree $n$ in $\K[X]$ in $\M(n)$ base field
operations, and that we can compute $f(g) \bmod h$ in $\CC(n)$
operations in $\K$, when $f,g,h$ are degree $n$ polynomials. We
additionally require that both $\M$ and $\CC$ are super-linear
functions, as in~\cite[Chapter~8]{GaGe03}, and that $\M(n)=O(\CC(n))$.
In particular, since we work in the monomial basis, multiplications
and inversions in $\F_q$ can be done in respectively $O(\M(n))$ and
$O(\M(n)\log(n))$ operations in $\F_p$, see again~\cite{GaGe03}.

The best known bound for $\M(n)$ is $O(n\log(n)\log\log(n))$, achieved
by using Fast Fourier Transform~\cite{Schonhage1971,CaKa91}.  The most
well-known bound for $\CC(n)$ is $O(n^{(\omega+1)/2})$, due to Brent
and Kung~\cite{BrKu78}, where $\omega$ is such that matrices of size
$n$ over any field $\K$ can be multiplied in $O(n^\omega)$ operations
in $\K$; this estimate assumes that $\omega > 2$, otherwise some
logarithmic terms may appear. Using the algorithm of Coppersmith and
Winograd~\cite{CoWi90}, we can take $\omega \le 2.37$ and thus
$\CC(n)=O(n^{1.69})$; an algorithm by Huang and Pan~\cite{HuPa98}
actually achieves a slightly better exponent of $1.67$, by means of
rectangular matrix multiplication.

\paragraph{Main result.}
We will focus in this paper on the case of $t$-th root extraction,
where $t$ is a prime divisor of $q-1$; the general case of $m$-th root
extraction, with $m$ arbitrary, can easily be reduced to this case
(see the discussion after Theorem~\ref{theo:main}).

The core of our algorithm is a reduction of $t$-th root extraction in
$\F_q$ to $t$-th root extraction in an extension of $\F_p$ of smaller
degree. Our algorithm is probabilistic of Las Vegas type, so 
its running time is given as an expected number of operations. With
this convention, our main result is the following.

\begin{theorem}\label{theo:main}
  Let $t$ be a prime factor of $q-1$, with $q=p^n$, and let $s$ be the
  order of $p$ in $\mathbb{Z}/t\mathbb{Z}$. Given $a\in \F_q^*$, one
  can decide if $a$ is a $t$-th power in $\F_q^*$, and if so compute
  one of its $t$-th roots, by means of the following operations:
  \begin{itemize}
  \item an expected $O(s\M(n)\log(p) + \CC(n)\log(n))$ operations in
    $\F_p$;
  \item a $t$-th root extraction in $\F_{p^s}$.
  \end{itemize}
\end{theorem}

Thus, we replace $t$-th root extraction in a degree $n$ extension by a
$t$-th root extraction in an extension of degree $s \le \min(n,t)$.
The extension degree $s$ is the largest one for which $t$ still
divides $p^s-1$, so iterating the process does not bring any
improvement: the $t$-th root extraction in $\F_{p^s}$ must be dealt
with by another algorithm. The smaller $s$ is, the better.

A useful special case is $t=2$, that is, we are taking square roots;
the assumption that $t$ divides $q-1$ is then satisfied for all odd
primes $p$ and all $n$. In this case, we have $s=1$, so the second
step amounts to square root extraction in $\F_p$. Since this can be
done in $O(\log(p))$ expected operations in $\F_p$, the total running
time of the algorithm is an expected $O(\M(n)\log(p) + \CC(n)\log(n))$
operations in $\F_p$. 

A previous algorithm by Kaltofen and Shoup~\cite{KaltofenShoup1997}
allows one to compute $t$-th roots in $\F_{p^n}$ in expected time
$O((\M(t)\M(n)\log(p) + t \CC(n) + \CC(t)\M(n))\log(n))$; we discuss
it further in the next section. This algorithm requires no assumption
on $t$, so it can be used in our algorithm in the case $s > 1$, for
$t$-th root extraction in $\F_{p^s}$. Then, its expected running time
is $O((\M(t)\M(s)\log(p) + t \CC(s) + \CC(t)\M(s))\log(s))$.

The strategy of using Theorem~\ref{theo:main} to reduce from $\F_q$ to
$\F_{p^s}$ then using the Kaltofen-Shoup algorithm over $\F_{p^s}$
is never more expensive than using the Kaltofen-Shoup algorithm
directly over $\F_{q}$. For $t=O(1)$, both strategies are within a
constant factor; but even for the smallest case $t=2$, our algorithm
has advantages (as explained in the last section). For larger $t$, the
gap in our favor will increase for cases when $s$ is small (such 
as when $t$ divides $p-1$, corresponding to $s=1$).

\medskip

Finally, let us go back to the remark above, that for any $m$, one can
reduce $m$-th root extraction of $a \in \F_q^*$ to computing $t$-th
roots, with $t$ dividing $q-1$; this is well known, see for
instance~\cite[Chapter~7.3]{BachSh1996}. We write $m = uv$ with $(v, q
- 1) = 1$ and $t \mid q - 1$ for every prime divisor $t$ of $u$, and
we assume that $a$ is indeed an $m$-th power.
\begin{itemize}
\item We first compute the $v$-th root $a_0$ of $a$ as
  $a_0=a^{v^{-1} \bmod q - 1}$ by computing the inverse $\ell$ of
  $v \bmod q-1$, and computing an $\ell$-th power in $\F_q$. This
  takes $O(n\M(n)\log(p))$ operations in $\F_p$.
\item Let $u = \prod_{i = 1}^d{m_i^{\alpha_i}}$ be the prime
  factorization of $u$, which we assume is given to us. Then, for $k =
  1, \dots, \alpha_1$, we compute an $m_1$-th root $a_k$ of $a_{k-1}$
  using Theorem~\ref{theo:main}, so that $a_{\alpha_1}$ is an
  $m_1^{\alpha_1}$-th root of $a_0$.

  One should be careful in the choice of the $m_1$-th roots (which are
  not unique), so as to ensure that each $a_k$ is indeed an
  $u/m_1^i$-th power: if the given $a_k$ is not such a power, we can
  multiply it by a $m_1$-th root of unity until we find a suitable
  one.  The root of unity can be found by the algorithm of
  Theorem~\ref{theo:main}.

  Once we know $a_{\alpha_1}$, the same process can be applied to
  compute an $m_2^{\alpha_2}$-th root of $a_{\alpha_1}$, and so on.
\end{itemize}

The first step, taking a root of order $v$, may actually be the
bottleneck of this scheme. When $v$ is small compared to $n$, it may
be better to use here as well the algorithm by Kaltofen and Shoup
mentioned above.


\paragraph{Organization of the paper.} 
The next section reviews and discusses known algorithms;
Section~\ref{section:newRootEx} gives the details of the root
extraction algorithm and some experimental results. In all the paper,
$(\F_q^*)^t$ denotes the set of $t$-th powers in $\F_q^*$.

\paragraph{Acknowledgments.} We thank NSERC and the Canada Research
Chairs program for financial support.


\section{Previous work}

Let $t$ be a prime factor of $q-1$.  In the rest of this section, we discuss
previous algorithms for $t$-th root extraction, with a special focus
on the case $t=2$ (square roots), which has attracted most attention
in the literature.

We shall see in \refsection{section:newRootEx} that given such a prime
$t$, the cost of testing for $t$-th power is always dominated by the
$t$-th root extraction; thus, for an input $a \in \F_q^*$, we always
assume that $a \in (\F_q^*)^t$.

All algorithms discussed below rely on some form of exponentiation in
$\F_q$, or in an extension of $\F_q$, with exponents that grow
linearly with $q$. As a result, a direct implementation using binary
powering uses $O(\log(q))$ multiplications in $\F_q$, that is,
$O(n\M(n)\log(p))$ operations in $\F_p$. Even using fast
multiplication, this is quadratic in $n$; alternative techniques
should be used to perform the exponentiation, when possible.

\paragraph{Some special cases of square root computation.}
If $G$ is a group with an odd order $s$, then the mapping $f: G
\rightarrow G$, $f(a) = a^2$ is an automorphism of $G$; hence, every
element $a \in G$ has a unique square root, which is $a^{(s +
  1)/2}$. Thus, if $q \equiv 3 \pmod 4$, the square root of any $a \in
(\F_q^*)^2$ is $a^{(q + 1) / 4}$; this is because $(\F_q^*)^2$ is a
group of odd order $(q - 1) / 2$.

More complex schemes allow one to compute square roots for some
increasingly restricted classes of prime powers $q$. The following
table summarizes results known to us; in each case, the algorithm uses
$O(1)$ exponentiations and $O(1)$ additions / multiplications in
$\F_q$. The table indicates what exponents are used in the
exponentiations.

  \begin{table}[h]
    \caption{Some special cases for square roots}\label{table1}
\begin{center}
\begin{tabular}{c|c|c}
Algorithm & $q$ & exponent \\\hline
folklore  & $3 \pmod 4$ & $(q + 1) / 4$ \\
Atkin     & $5 \pmod 8$ & $(q-5)/8$ \\
M\"uller~\cite{Muller04} & $9 \pmod {16}$ & $(q-1)/4$ and $(q-9)/16$\\
Kong {\it et al.}~\cite{KoCaYuLi06} & $9 \pmod {16}$ & $(q-9)/8$ and $(q-9)/16$
\end{tabular}
\end{center}
  \end{table}
\noindent As was said above, using a direct binary powering approach to
exponentiation, all these algorithms use $O(n\M(n)\log(p))$ operations
in $\F_p$.

\paragraph{Cipolla's square root algorithm.} 
To compute the square root of $a\in (\F_q^*)^2$, Cipolla's algorithm
uses an element $b$ in $\F_q$ such that $b^2 - 4a$ is not a square in
$\F_q$. Then, the polynomial $f(Y) = Y^2 - bY + a$ is irreducible over
$\F_q$, hence $\K=\F_q[Y]/\langle f\rangle$ is a field. Let $y$ be the
residue class of $Y$ modulo $\langle f\rangle$.  Since $f$ is the
minimal polynomial of $y$ over $\F_q$, $N_{\K/\F_q}(y) = a$, ensuring
that $\sqrt{a}=Y^{(q + 1) / 2} \bmod (Y^2 - bY + a)$.

Finding a quadratic nonresidue of the form $b^2 - 4a$ by choosing a
random $b \in \F_q$ requires an expected $O(1)$ attempts~\cite[page
  158]{BachSh1996}. The quadratic residue test, and the norm
computation take $O(\M(n)\log(n)+\log(p))$ and $O(n\M(n)\log(p))$
multiplications in $\F_p$ respectively. Therefore, the cost of the
algorithm is an expected $O(n\M(n)\log(p))$ operations in $\F_p$.

Algorithms extending Cipolla's to the computation of $t$-th roots in $\F_p$,
where $t$ is a prime factor of $p-1$, are
in~\cite{Williams72,WiHa93,NiHaSuKu09}.

\paragraph{The Tonelli-Shanks algorithm.}
We will describe the algorithm in the case of square roots, although
the ideas extend to higher orders. Tonelli's
algorithm~\cite{Tonelli1891} and Shanks' improvement~\cite{Shanks1972}
use discrete logarithms to reduce the problem to a subgroup of
$\F_q^*$ of odd order. Let $q - 1 = 2^r\ell$ with $(\ell, 2) = 1$ and
let $H$ be the unique subgroup of $\F_q^*$ of order $\ell$. Assume we
find a quadratic nonresidue $g \in \F_q^*$; then, the square root of
$a \in \F_q^*$ can be computed as follows: we can express $a$ as $g^sh
\in g^sH$ by solving a discrete logarithm in $\F_q^* / H$; $s$ is
necessarily even, so that $\sqrt{a} = g^{s / 2}h^{(\ell + 1) / 2}$.

According to \cite{Pohlig1978}, the discrete logarithm requires
$O(r^2\M(n))$ multiplications in $\F_p$; all other steps take
$O(n\M(n)\log(p))$ operations in $\F_p$. Hence, the expected running
time of the algorithm is $O((r^2 + n\log(p))\M(n))$ operations in
$\F_p$. Thus, the efficiency of this algorithm depends on the
structure of $\F_q^*$: there exists an infinite sequence of primes for
which the cost is $O(n^2\M(n)\log(p)^2)$, see~\cite{Tornaria2002}.


\paragraph{Improving the exponentiation.}
All algorithms seen before use at best $O(n\M(n)\log(p))$ operations
in $\F_p$, because of the exponentiation. Using ideas going back
to~\cite{ItTs88}, Barreto {\it et al.}~\cite{BaKiLySc02} observed that
for some of the cases seen above, the exponentiation can be improved,
giving a cost subquadratic in $n$.

For instance, when taking square roots with $q = 3 \pmod 4$, the
exponentiation $a^{(q+1)/4}$ can be reduced to computing
$a^{1+u+\cdots+u^{(n-3)/2}}$, with $u=p^2$, plus two (cheap)
exponentiations with exponents $p(p-1)$ and $(p+1)/4$. The special
form of the exponent $1+u+\cdots+u^{(n-3)/2}$ makes it possible to
apply a binary powering approach, involving $O(\log(n))$
multiplications and exponentiations, with exponents that are powers of
$p$.

Further examples for square roots are discussed
in~\cite{KoCaYuLi06,HaChKi09}, covering the entries of
Table~\ref{table1}. These references assume a normal basis
representation; using (as we do) the monomial basis and modular
composition techniques (which will be explained in the next section),
the costs become $O(\M(n)\log(p)+\CC(n)\log(n))$. Some cases of higher
index roots are in~\cite{BaVo06}: if $t$ is a factor of $p-1$, such
that $t^2$ does not divide $p-1$, and if $\gcd(n,t)=1$, then $t$-th
root extraction can be done using $O(t\M(n)\log(p)+\CC(n)\log(n))$
operations in $\F_p$.

\paragraph{Kaltofen and Shoup's algorithm.}
Finally, we mention what is, as far as we know, the only algorithm
achieving an expected subquadratic running time in $n$ (using the
monomial basis representation), without any assumption on $p$.

Consider a factor $t$ of $q-1$. To compute a $t$-th root of
$a\in(\F_q^*)^t$, the idea is simply to factor $Y^t - a \in \F_q[Y]$
using polynomial factorization techniques. Since we know that $a$ is a
$t$-th power, this polynomial splits into linear factors, so we
can use an Equal Degree Factorization (EDF) algorithm.

A specialized EDF algorithm, dedicated to the case of high-degree
extension of a given base field, was proposed by Kaltofen and Shoup
\cite{KaltofenShoup1997}. It mainly reduces to the computation of a
trace-like quantity $b+b^p + \cdots + b^{p^{n-1}}$, where $b$ is a
random element in $\F_q[Y]/\langle Y^t-a\rangle$. Using a binary
powering technique similar to the one of the previous paragraph, this
results in an expected running time of $O((\M(t)\M(n)\log(p) + t
\CC(n) + \CC(t)\M(n))\log(n))$ operations in $\F_p$; remark that this
estimate is faster than what is stated in~\cite{KaltofenShoup1997} by
a factor $\log(t)$, since here we only need one root, instead of the
whole factorization. In the particular case $t=2$, this becomes
$O((\M(n)\log(p) + \CC(n))\log(n))$. This achieves a running time
subquadratic in~$n$.

This idea actually allows one to compute a $t$-th root, for arbitrary
$t$: starting from the polynomial $Y^t-a$, we apply the above
algorithm to $\gcd(Y^t-a, Y^q-Y)$; computing $Y^q$ modulo $Y^t-a$ can
be done by the same binary powering techniques.



\section{A new root extraction algorithm}\label{section:newRootEx}

In this section, we focus on $t$-th root extraction in $\F_q$, for $t$
a prime dividing $q-1$ (as we saw in \refsection{section:intro},
$m$-th root extraction, for an integer $m \ge 2$, reduces to taking 
$t$-th roots, where $t$ is a prime factor of $m$ dividing $q-1$).

The algorithm we present uses the trace $\F_q \to \F_{q'}$, for some
subfield $\F_{q'} \subset \F_q$ to reduce $t$-th root extraction in
$\F_q$ to $t$-th root extraction in $\F_{q'}$. We assume as before
that the field $\F_q$ is represented by a quotient $\F_p[X] / \langle
f\rangle$, with $f(X) \in \F_p[X]$ a monic irreducible polynomial of
degree $n$. We let $x$ be the residue class of $X$ modulo $\langle
f\rangle$.

Since we will handle both $\F_q$ and $\F_{q'}$, conversions may be
needed. We recall that the minimal polynomial $g \in \F_p[Z]$ of an
element $b \in \F_q$ can be computed in $O(\CC(n)+\M(n)\log(n))$
operations in $\F_p$~\cite{Shoup1994}. Then, $\F_{q'}=\F_p[Z]/\langle
g\rangle$ is a subfield of $\F_q=\F_p[X]/\langle f\rangle $; given $r \in
\F_{q'}$, written as a polynomial in $Z$, we obtain its representation
on the monomial basis of $\F_q$ by means of a modular composition, in
time $\CC(n)$. We will write this operation ${\rm Embed}(r,\F_q)$.
Note that when $b$ is in $\F_p$, all these operations are actually
free.


\subsection{An auxiliary algorithm}

We first discuss a binary powering algorithm to solve the following
problem. Starting from $\lambda \in \F_q$, we are going to compute
$$
\alpha_i(\lambda) = \lambda^{1 + p^s} + \lambda^{1 + p^s + p^{2s}} + \cdots + \lambda^{1 + p^s + p^{2s} + \cdots + p^{is}}
$$ for given integers $i, s > 0$. This question is similar to (but
distinct from) some exponentiations and trace-like computations we
discussed before; our solution will be a similar binary powering 
approach, which will perform $O(\log(i))$ multiplications and
exponentiations by powers of $p$. Let
$$
\xi_i = x^{p^{is}},
\quad 
\zeta_i(\lambda) = \lambda^{p^s + p^{2s} + \cdots + p^{is}} 
\quad\text{and}\quad 
\delta_i(\lambda) = \lambda^{p^s} + \lambda^{p^s + p^{2s}} + \cdots + \lambda^{p^s + p^{2s} + \cdots + p^{is}},
$$ where all quantities are computed in $\F_q$, that is, modulo $f$;
for simplicity, in this paragraph, we will write $\alpha_i$, $\zeta_i$
and $\delta_i$. Note that $\alpha_i=\lambda \delta_i$, and that we
have the following relations:
$$\xi_1 = x^{p^s}, \quad \zeta_1 = \lambda^{p^s}, \quad \delta_1 = \lambda^{p^s}$$
and
$$
\xi_i =
\begin{cases}
\xi_{i / 2}^{p^{is / 2}} & \text{if $i$ is even}  \\
\xi_{i - 1}^{p^s} & \text{if $i$ is odd,}
\end{cases} \quad
\zeta_i = 
\begin{cases}
\zeta_{i / 2}\zeta_{i / 2}^{p^{is / 2}} & \text{if $i$ is even}  \\
\zeta_1\zeta_{i - 1}^{p^s} & \text{if $i$ is odd,}
\end{cases} \quad
\delta_i = 
\begin{cases}
\delta_{i / 2} + \zeta_{i / 2}\delta_{i / 2}^{p^{is / 2}} & \text{if $i$ is even}  \\
\delta_{i-1} + \zeta_i & \text{if $i$ is odd.}
\end{cases}$$
Because we are working in a monomial basis, computing exponentiations
to powers of $p$ is not a trivial task; we will perform them using the
following modular composition technique from~\cite{GaSh92}.

Take $j \ge 0$ and $r\in \F_q$, and let $R$ and $\Xi_j$ be the
canonical preimages of respectively $r$ and $\xi_j$ in $\F_p[X]$;
then, we have
$$r^{p^{js}} = R(\Xi_j) \bmod f,$$ see for instance~\cite[Chapter
  14.7]{GaGe03}. We will simply write this as $r(\xi_j)$, and note
that it can be computed using one modular composition, in time
$\CC(n)$. These remarks give us the following recursive algorithm,
where we assume that $\xi_1=x^{p^s}$ and $\zeta_1=\lambda^{p^s}$ are already
known.

\begin{algorithm}
[XiZetaDelta$(\lambda,i,\xi_1,\zeta_1)$]
\label{algorithm:xizetadelta}
\begin{algorithmic}[1]
\REQUIRE $\lambda$, a positive integer $i$, $\xi_1=x^{p^s}$, $\zeta_1=\lambda^{p^s}$
\ENSURE $\xi_i$, $\zeta_i$, $\delta_i$
\IF {$i=1$} 
\RETURN $\xi_1$, $\zeta_1$, $\zeta_1$
\ENDIF
\STATE $j \leftarrow \lfloor i/2\rfloor$
\STATE $\xi_{j},\zeta_{j},\delta_j \leftarrow {\rm XiZetaDelta}(\lambda,j,\xi_1,\zeta_1)$ 
\STATE $\xi_{2j} \leftarrow \xi_j(\xi_j)$
\STATE $\zeta_{2j} \leftarrow \zeta_j\cdot \zeta_j(\xi_j)$
\STATE $\delta_{2j}\leftarrow \delta_j+\zeta_j \delta_j(\xi_j)$
\IF {$i$ is even} 
\RETURN $\xi_{2j}$, $\zeta_{2j}$, $\delta_{2j}$
\ENDIF
\STATE $\xi_i \leftarrow \xi_{2j}(\xi_1)$
\STATE $\zeta_i \leftarrow \zeta_1\cdot \zeta_{2j}(\xi_1)$
\STATE $\delta_i \leftarrow \delta_{2j}+\zeta_i$
\RETURN $\xi_i$, $\zeta_i$, $\delta_i$
\end{algorithmic}
\end{algorithm}
We deduce the following algorithm for computing $\alpha_i(\lambda)$.

\begin{algorithm}
[Alpha$(\lambda,i)$]
\label{algorithm:alpha}
\begin{algorithmic}[1]
\REQUIRE $\lambda$, a positive integer $i$
\ENSURE $\alpha_i$
\STATE $\xi_1 \leftarrow x^{p^s}$
\STATE $\zeta_1 \leftarrow \lambda^{p^s}$
\STATE $\xi_i$, $\zeta_i$, $\delta_i \leftarrow {\rm XiZetaDelta}(\lambda,i,\xi_1,\zeta_1)$ 
\RETURN $\lambda \delta_i$
\end{algorithmic}
\end{algorithm}

\begin{proposition}
  \refalgorithm{algorithm:alpha} computes $\alpha_i(\lambda)$ using
  $O(\CC(n)\log(is) + \M(n)\log(p))$ operations in $\F_p$.
\end{proposition}
\begin{proof}
  To compute $x^{p^s}$ and $\lambda^{p^s}$ we first compute $x^p$ and
  $\lambda^p$ using $O(\log(p))$ multiplications in $\F_q$, and then
  do $O(\log (s))$ modular compositions modulo $f$. The depth of the
  recursion in Algorithm~\ref{algorithm:xizetadelta} is $O(\log(i))$;
  each recursive call involves $O(1)$ additions, multiplications and
  modular compositions modulo $f$, for a total time of $O(\CC(n))$ per
  recursive call.
\end{proof}

As said before, the algorithm can also be implemented using a normal
basis representation. Then, exponentiations to powers of $p$ become
trivial, but multiplication becomes more difficult. We leave these
considerations to the reader.


\subsection{Taking $t$-th roots} 

We will now give our root extraction algorithm. As said before, we now
let $t$ be a prime factor of $q-1$, and we let $s$ be the order of $p$
in $\mathbb{Z}/t\mathbb{Z}$. Then $s$ divides $n$, say $n = s\ell$.

We first explain how to test for $t$-th powers. Testing whether $a
\in \F_q^*$ is a $t$-th power is equivalent to testing whether
$a^{(q-1)/t}=1$. Let $\zeta = a^{(p^s - 1) / t}$; then $a^{(q - 1) /
  t} = \zeta^{1 + p^s + \cdots + p^{(\ell - 1)s}}$. Computing $\zeta$
requires $O(s\M(n)\log(p))$, and computing $\zeta^{1 + p^s + \cdots +
  p^{(\ell - 1)s}}$ using \refalgorithm{algorithm:xizetadelta}
requires $O(\CC(n)\log(n) + \M(n)\log(p))$ operations in
$\F_p$. Therefore, testing for a $t$-th power takes $O(\CC(n)\log(n) +
s\M(n)\log(p))$ operations in $\F_p$.

In the particular case when $t$ divides $p - 1$, we can actually do
better: we have $a^{(q - 1) / t} = \operatorname{res}(f, a)^{(p - 1) /
  t}$, where $\operatorname{res}(\cdot, \cdot)$ is the resultant
function. The resultant can be computed using $O(\M(n)\log(n))$
operations in $\F_p$, so the whole test can be done using
$O(\M(n)\log(n)+\log(p))$ operations in $\F_p$.

In any case, we can now assume that we are given $a \in (\F_q^*)^t$,
with $t$-th root $\gamma \in \F_q$. 
Defining $\beta = T_{\F_q / \F_{q'}}(\gamma)$, where
$T_{\F_q / \F_{q'}}:\F_q \to \F_{q'}$ is the trace linear form and
$q' = p^s$, we have
\begin{align}
\label{equation:tr-square}
\beta 
& = \sum_{i = 0}^{\ell - 1} \gamma^{p^{is}} \nonumber \\
& = \gamma(1 + \gamma^{p^s - 1} + \gamma^{p^{2s} - 1} + \cdots + \gamma^{p^{(\ell - 1)s} - 1}) \nonumber \\
& = \gamma(1 + a^{(p^s - 1) / t} + a^{(p^{2s} - 1) / t} + \cdots + a^{(p^{(\ell - 1)s} -1) / t}).
\end{align}
Let $b = 1 + a^{(p^s - 1) / t} + a^{(p^{2s} - 1) / t} + \cdots +
a^{(p^ {(\ell - 1)s} -1) / t}$, so that
\refequation{equation:tr-square} gives $\beta = \gamma b$.  Taking the
$t$-th power in both sides results in the equation $\beta^t = ab^t$
over $\F_{q'}$. Since we know $a$, and we can compute $b$, we can thus
determine $\beta$ by $t$-th root extraction in $\F_{q'}$. Then, if we
assume that $b \ne 0$ (or equivalently that $\beta \ne 0$), we deduce
$\gamma = \beta b^{-1}$; to resolve the issue that $\beta$ may be
zero, we will replace $a$ by $a'=ac^t$, for a random element $c \in
\F_q^*$.

Computing the $t$-th root of $a'b^t$ in $\F_{q'}$ is done as follows.
We first compute the minimal polynomial $g\in\F_p[Z]$ of $a'b^t$, and
let $z$ be the residue class of $Z$ in $\F_p[Z]/\langle g \rangle$.
Then, we compute a $t$-th root $r$ of $z$ in $\F_p[Z]/\langle g \rangle$, and
embed $r$ in $\F_q$. The computation of $r$ is done by a black-box
$t$-th root extraction algorithm, denoted by $r \mapsto r^{1/t}$.

It remains to explain how to compute $b$ efficiently. Let $\lambda =
a^{(p^s - 1) / t}$; then, one verifies that $b = 1 + \lambda +
\alpha_{\ell - 2}(\lambda)$, so we can use the algorithm of the
previous subsection. Putting all together, we obtain the following
algorithm

\begin{algorithm}
[$t$-th root in $\F_q^*$]
\label{algorithm:new}
\begin{algorithmic}[1]
\REQUIRE $a \in (\F_q^*)^t$
\ENSURE a $t$-th root of $a$
\STATE $s \leftarrow $ the order of $p$ in $\mathbb{Z}/t\mathbb{Z}$
\STATE $\ell \leftarrow n / s$
\REPEAT
\STATE choose a random $c \in \F_q$
\STATE $a'\leftarrow ac^t$
\STATE $\lambda \leftarrow {a'}^{(p^s-1)/t}$
\STATE $b \leftarrow 1+\lambda+{\rm Alpha}(\lambda,\ell-2)$
\UNTIL {$b \ne 0$}
\STATE $g \leftarrow {\rm MinimalPolynomial}(a'b^t)$
\STATE $\beta \leftarrow z^{1/t}$ in $\F_q[Z]/\langle g\rangle$
\RETURN ${\rm Embed}(\beta,\F_q) b^{-1}c^{-1}$
\end{algorithmic}
\end{algorithm}

The following proposition proves Theorem~\ref{theo:main}.

\begin{proposition}
  \refalgorithm{algorithm:new} computes a $t$-th root of $a$ using an
  expected $O(s\M(n)\log(p) + \CC(n)\log(n))$ operations in $\F_p$,
  plus a $t$-th root extraction in $\F_{q'}$.
\end{proposition}
\begin{proof}
  Note first that $\beta=0$ means that $T_{\F_q / \F_{q'}}(\gamma
  c)=0$. There are $q/q'$ values of $c$ for which this is the case, so
  we expect to have to choose $O(1)$ elements in $\F_q$ before exiting
  the repeat \dots until loop. Each pass in the loop uses
  $O(s\M(n)\log(p))$ operations in $\F_p$ to compute $a'$ and
  $\lambda$, and $O(\CC(n)\log(n) + \M(n)\log(p))$ operations in $\F_p$
  to compute $b$.

  Given $a'$ and $b$, one obtains $b^t$ and $a'b^t$ using another
  $O(s\M(n)\log(p))$ operations in $\F_p$; then, computing $g$ takes
  time $O(\CC(n) + \M(n)\log(n))$. 

  After the black-box call to $t$-th root extraction modulo $g$,
  embedding $\beta$ in $\F_q$ takes time $\CC(n)$. We can then deduce
  $\gamma$ by two divisions in $\F_q$, using $O(\M(n)\log(n))$
  operations in $\F_p$; this is negligible compared to the cost of all
  modular compositions.
\end{proof}



\subsection{Experimental results}
We have implemented our root extraction algorithm, in the case $m=2$
(that is, we are taking square roots); our implementation is based on
Shoup's NTL~\cite{NTL2009}. Figure~\ref{figure:sqrtTiming} compares
our algorithm to Cipolla's and Tonelli-Shanks' algorithms over $\F_q$,
with $q = p^n$, for the randomly selected prime $p =
\seqsplit{348975609381470925634534573457497}$, and different values of
the extension degree~$n$. 

\begin{figure}[ht]
\begin{center}
\includegraphics[width = 9cm]{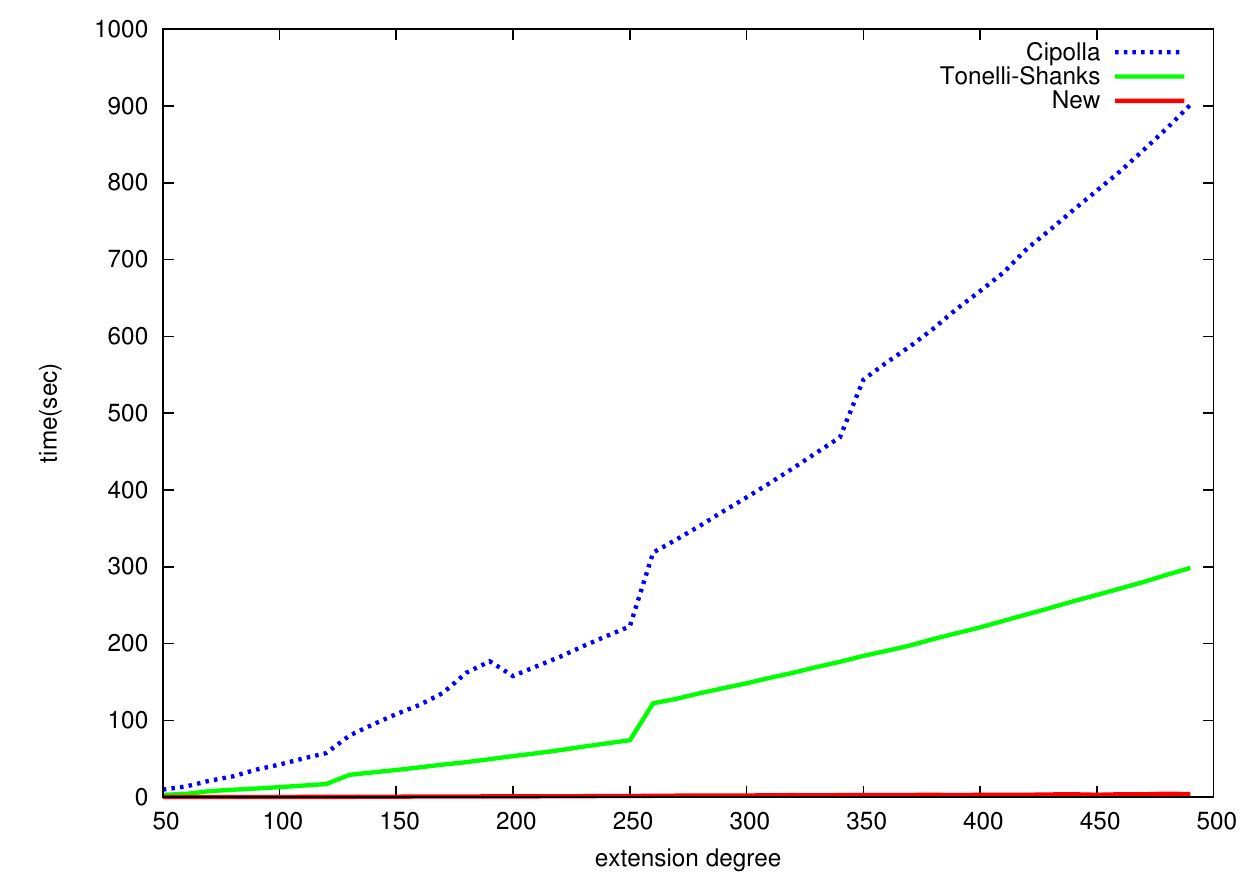}
\end{center}
\caption{\small Our square root algorithm vs. Cipolla's and
  Tonelli-Shanks' algorithms.}
\label{figure:sqrtTiming}
\end{figure}

Remember that the bottleneck in Cipolla's and Tonelli-Shanks'
algorithms is the exponentiation, which takes $O(n\M(n)\log(p))$
operation in $\F_p$. As it turns out, NTL's implementation of modular
composition has $\omega=2$; this means that with this implementation
we have $\CC(n)=O(n^2)$, and our algorithm takes expected time
$O(\M(n)\log(p)+n^2 \log(n))$. Although this implementation is not
subquadratic in $n$, it remains faster than Cipolla's and
Tonelli-Shanks' algorithms, in theory and in practice.

Next, Figure~\ref{figure:sqrtTimingKvN} compares our NTL
implementation of the EDF algorithm proposed by Kaltofen and Shoup,
and our square root algorithm (note that the range of reachable degrees
is much larger that in the first figure). We have ran the algorithms
for several random elements for each extension degree. The vertical dashed
lines, and the green line show the running time range, and the average 
running time of Kaltofen and Shoup algorithm respectively. In the case of 
our algorithm (the red graph), the vertical ranges are invisible because the 
deviation from the average is $\approx 10^{-2}$ seconds. 

\begin{figure}[ht]
\begin{center}
\includegraphics[width = 9cm]{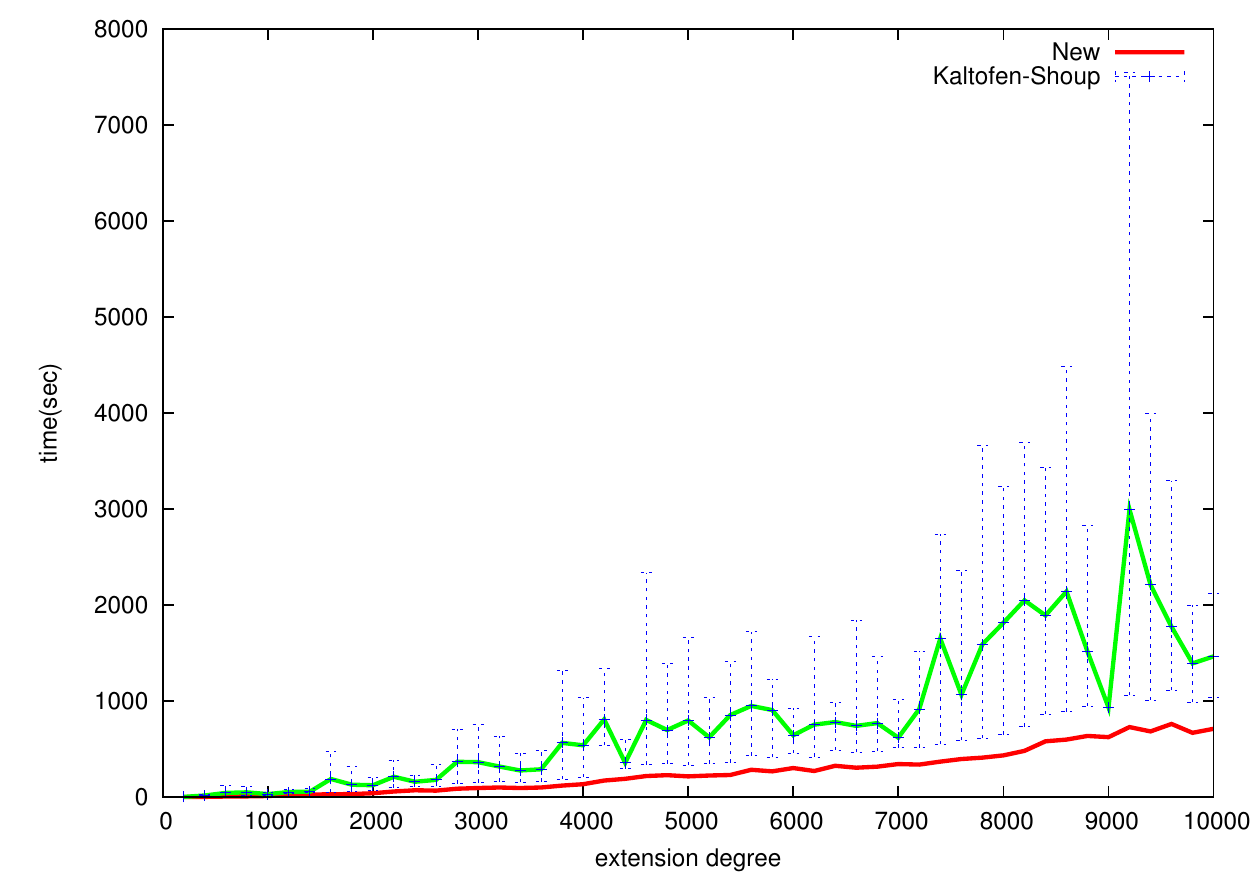}
\end{center}
\caption{\small Our algorithm vs.\ Kaltofen and Shoup's algorithm.}
\label{figure:sqrtTimingKvN}
\end{figure}

This time, the results are closer. Nevertheless, it appears that the
running time of our algorithm behaves more ``smoothly'', in the sense that
random choices seem to have less influence. This is indeed the
case. The random choice in Kaltofen and Shoup's algorithm succeeds
with probability about $1/2$; in case of failure, we have to run to
whole algorithm again. In our case, our choice of an element $c$ in
$\F_q^*$ fails with probability $1/p \ll 1/2$; then, there is still
randomness involved in the $t$-th root extraction in $\F_p$, but this
step was negligible in the range of parameters where our experiments
were conducted. 

Another way to express this is to compare the standard deviations in
the running times of both algorithms. In the case of Kaltofen-Shoup's
algorithm, the standard deviation is about $1/\sqrt{2}$ of the average
running time of the whole algorithm. For our algorithm, the standard
deviation is no more than $1/\sqrt{p}$ of the average running time of the
trace-like computation (which is the dominant part), plus $1/\sqrt{2}$
of the average running time of the root extraction in $\F_p$ (which is cheap).


\bibliographystyle{plain}
\bibliography{references}

\end{document}